\documentstyle[12pt]{article}

\font\twlgot =eufm10 scaled \magstep1 \font\egtgot =eufm8
\font\sevgot =eufm7 \font\twlmsb =msbm10 scaled \magstep1
\font\egtmsb =msbm8 \font\sevmsb =msbm7

\newfam\gotfam

\textfont\gotfam\twlgot \scriptfont\gotfam\egtgot
\scriptscriptfont\gotfam\sevgot

\newfam\msbfam
\textfont\msbfam\twlmsb \scriptfont\msbfam\egtmsb
\scriptscriptfont\msbfam\sevmsb
\def\Bbb{\protect\pBbb}
\def\pBbb{\relax\ifmmode\expandafter\Bb\else\typeout{You cann't use
Bbb in text mode}\fi}
\def\Bb #1{{\fam\msbfam\relax#1}}

\def\thebibliography#1{\section*{References}\list
  {[\arabic{enumi}]}{\settowidth\labelwidth{#1}\leftmargin\labelwidth
    \advance\leftmargin\labelsep
    \usecounter{enumi}}
    \def\newblock{\hskip .11em plus .33em minus .07em}
    \sloppy\clubpenalty4000\widowpenalty4000
    \sfcode`\.=1000\relax}

\def\op#1{\mathop{\fam0 #1}\limits}

\newcommand{\id}{{\rm Id\,}}

\newcommand{\beq}{\begin{equation}}
\newcommand{\eeq}{\end{equation}}
\newcommand{\ben}{\begin{eqnarray}}
\newcommand{\een}{\end{eqnarray}}
\newcommand{\be}{\begin{eqnarray*}}
\newcommand{\ee}{\end{eqnarray*}}
\newcommand{\bea}{\begin{eqalph}}
\newcommand{\eea}{\end{eqalph}}

\newcommand{\cV}{{\cal V}}

\newcommand{\cH}{{\cal H}}

\newcommand{\cF}{{\cal F}}

\newcommand{\cG}{{\cal G}}

\newcommand{\bs}{{\bf s}}

\newcommand{\al}{\alpha}

\newcommand{\bt}{\beta}
\newcommand{\dl}{\delta}
\newcommand{\la}{\lambda}

\newcommand{\f}{\phi}

\newcommand{\Om}{\Omega}

\newcommand{\m}{\mu}

\newcommand{\g}{\gamma}

\newcommand{\vt}{\vartheta}
\newcommand{\vf}{\varphi}
\newcommand{\up}{\upsilon}

\newcommand{\di}{{\rm dim\,}}

\newcommand{\si}{\sigma}

\newcommand{\w}{\wedge}

\newcommand{\dr}{\partial}
\newcommand{\ar}{\op\longrightarrow}
\newcommand{\ot}{\otimes}

\newcommand{\ve}{\varepsilon}

\let\ssection=\section
\renewcommand{\section}{\setcounter{equation}{0}\ssection}

\newcounter{example}[section]
\newcounter{remark}[section]
\newcounter{theorem}[section]
\newcounter{proposition}[section]
\newcounter{lemma}[section]
\newcounter{corollary}[section]
\newcounter{definition}[section]

\def\theremark{\arabic{section}.\arabic{remark}}

\def\thedefinition{\arabic{section}.\arabic{definition}}

\newenvironment{proof}{\noindent
{\bf Proof.}}{$\Box$ \medskip}
\newenvironment{rem}{\refstepcounter{remark}\medskip\noindent{\bf
Remark \theremark.}}{\medskip}
\newenvironment{theo}{\refstepcounter{definition}
\bigskip\noindent{\bf Theorem \thedefinition.} \it}{\medskip}
\newenvironment{prop}{\refstepcounter{definition}
\bigskip\noindent{\bf Proposition \thedefinition.}\it}{\medskip}
\newenvironment{lem}{\refstepcounter{definition}
\bigskip\noindent{\bf Lemma \thedefinition.}\it}{\medskip}

\newenvironment{defi}{\refstepcounter{definition}
\bigskip\noindent{\bf Definition \thedefinition.}\it}{\medskip}

\newcommand{\mar}[1]{}

\begin{document}
\hbox{}

\begin{center}

{\large \bf SUPERINTEGRABLE HAMILTONIAN SYSTEMS WITH NONCOMPACT
INVARIANT SUBMANIFOLDS. KEPLER SYSTEM}
\bigskip

{\sc G. SARDANASHVILY}
\bigskip

\begin{small}

{\it Department of Theoretical Physics, Physics Faculty, Moscow
State University

117234 Moscow, Russia}
\end{small}

\end{center}

\bigskip

\begin{small}

\noindent The Mishchenko--Fomenko theorem on superintegrable
Hamiltonian systems is generalized to superintegrable Hamiltonian
systems with noncompact invariant submanifolds. It is formulated
in the case of globally superintegrable Hamiltonian systems which
admit global generalized action-angle coordinates. The well known
Kepler system falls into two different globally superintegrable
systems with compact and noncompact invariant submanifolds.

\medskip

\end{small}

\section{Introduction}

Let $(Z,\Om)$ be a $2n$-dimensional connected symplectic manifold.
Given a superintegrable system
\mar{i00}\beq
F=(F_1,\ldots,F_k), \qquad n\leq k<2n, \label{i00}
\eeq
on $(Z,\Om)$ (Definition \ref{i0}), the well known Mishchenko --
Fomenko theorem (Theorem \ref{nc0}) states the existence of
(semi-local) generalized action-angle coordinates around its
connected compact invariant submanifold
\cite{bols03,fasso05,mishc}. If $k=n$, this is the case of
completely integrable systems (Definition \ref{cmp21}).

The Mishchenko -- Fomenko theorem has been extended to
superintegrable systems with noncompact invariant submanifolds
(Theorem \ref{nc0'}) \cite{fior2}. These submanifolds are
diffeomorphic to a toroidal cylinder
\mar{g120}\beq
\Bbb R^{m-r}\times T^r, \qquad m=2n-k, \qquad 0\leq r\leq m.
\label{g120}
\eeq
Partially and completely integrable systems with noncompact
invariant submanifolds were studied in \cite{fior2,jmp03,vin}. Our
goal here is the following.

We formulate Theorem \ref{nc0'} in the case of globally
superintegrable Hamiltonian systems, which admit global
generalized action-angle coordinates (Definition \ref{cmp30},
Theorem \ref{cmp36}). Herewith, Theorem \ref{cmp34} establishes
the sufficient condition of the existence of global generalized
action-angle coordinates \cite{jmp07} (see \cite{daz,fasso05} for
the case of compact invariant submanifolds).

Note that the Mishchenko -- Fomenko theorem is mainly applied to
superintegrable systems whose integrals of motion form a compact
Lie algebra. The group generated by flows of their Hamiltonian
vector fields is compact. Since a fibration of a compact manifold
possesses compact fibers, invariant submanifolds of such a
superintegrable system are compact. With Theorems \ref{nc0'} and
\ref{cmp36}, one can describe superintegrable Hamiltonian system
with an arbitrary Lie algebra of integrals of motion.

It may happen that a Hamiltonian system falls into different
superintegrable Hamiltonian systems on different open subsets of a
symplectic manifold. This is just the case of the Kepler system
considered in Section 7. It contains two different globally
superintegrable systems on different open subsets of a phase space
$\Bbb R^4$. Their integrals of motion form the Lie algebras
$so(3)$ and $so(2,1)$ with compact and non-compact invariant
submanifolds, respectively.

\section{The Mishchenko -- Fomenko theorem in a general setting}

Throughout the paper, all functions and maps are smooth, and
manifolds are real smooth and paracompact. We are not concerned
with the real-analytic case because a paracompact real-analytic
manifold admits the partition of unity by smooth functions. As a
consequence, sheaves of modules over real-analytic functions need
not be acyclic that is essential for our consideration.

\begin{defi} \label{i0} \mar{i0}
Let $(Z,\Om)$ be a $2n$-dimensional connected symplectic manifold,
and let $(C^\infty(Z), \{,\})$ be the Poisson algebra of smooth
real functions on $Z$. A subset $F$ (\ref{i00}) of the Poisson
algebra $C^\infty(Z)$ is called a superintegrable system if the
following conditions hold.

(i) All the functions $F_i$ (called the generating functions of a
superintegrable system) are independent, i.e., the $k$-form
$\op\w^kdF_i$ nowhere vanishes on $Z$. It follows that the map
$F:Z\to \Bbb R^k$ is a submersion, i.e.,
\mar{nc4}\beq
F:Z\to N=F(Z) \label{nc4}
\eeq
is a fibered manifold over a domain (i.e., contractible open
subset) $N\subset\Bbb R^k$ endowed with the coordinates $(x_i)$
such that $x_i\circ F=F_i$.

(ii) There exist smooth real functions $s_{ij}$ on $N$ such that
\mar{nc1}\beq
\{F_i,F_j\}= s_{ij}\circ F, \qquad i,j=1,\ldots, k. \label{nc1}
\eeq

(iii) The matrix function $\bs$ with the entries $s_{ij}$
(\ref{nc1}) is of constant corank $m=2n-k$ at all points of $N$.
\end{defi}

If $k=n$, then $\bs=0$, and we are in the case of completely
integrable systems as follows.

\begin{defi} \label{cmp21} \mar{cmp21} The subset $F$, $k=n$, (\ref{i00})
of the Poisson algebra $C^\infty(Z)$ on a symplectic manifold
$(Z,\Om)$ is called a completely integrable system if $F_i$ are
independent functions in involution.
\end{defi}

If $k>n$, the matrix $\bs$ is necessarily nonzero. Therefore,
superintegrable systems also are called noncommutative completely
integrable systems. If $k=2n-1$, a superintegrable system is
called maximally superintegrable.

The following two assertions clarify the structure of
superintegrable systems \cite{fasso05,fior2}.

\begin{prop} \label{nc7} \mar{nc7} Given a symplectic manifold $(Z,\Om)$,
let $F:Z\to N$ be a fibered manifold such that, for any two
functions $f$, $f'$ constant on fibers of $F$, their Poisson
bracket $\{f,f'\}$ is so. Then $N$ is provided with an unique
coinduced Poisson structure $\{,\}_N$ such that $F$ is a Poisson
morphism \cite{vaism}.
\end{prop}

Since any function constant on fibers of $F$ is a pull-back of
some function on $N$, the superintegrable system (\ref{i00})
satisfies the condition of Proposition \ref{nc7} due to item (ii)
of Definition \ref{i0}. Thus, the base $N$ of the fibration
(\ref{nc4}) is endowed with a coinduced Poisson structure of
corank $m$. With respect to coordinates $x_i$ in item (i) of
Definition \ref{i0} its bivector field reads
\mar{cmp1}\beq
w=s_{ij}(x_k)\dr^i\w\dr^j. \label{cmp1}
\eeq

\begin{prop} \label{nc8} \mar{nc8} Given a fibered manifold $F:Z\to N$ in Proposition \ref{nc7},
the following conditions are equivalent \cite{fasso05,libe}:

(i) the rank of the coinduced Poisson structure $\{,\}_N$ on $N$
equals $2\di N-\di Z$,

(ii) the fibers of $F$ are isotropic,

(iii) the fibers of $F$ are  maximal integral manifolds of the
involutive distribution spanned by the Hamiltonian vector fields
of the pull-back $F^*C$ of Casimir functions $C$ of the coinduced
Poisson structure (\ref{cmp1}) on $N$.
\end{prop}

It is readily observed that the fibered manifold $F$ (\ref{nc4})
obeys condition (i) of Proposition \ref{nc8} due to item (iii) of
Definition \ref{i0}, namely, $k-m= 2(k-n)$.

Fibers of the fibered manifold $F$ (\ref{nc4}) are called the
invariant submanifolds.

\begin{rem} \label{cmp8} \mar{cmp8} In many physical models, condition (i) of Definition
\ref{i0} fails to hold. Often, it is replaced with that a subset
$Z_R\subset Z$ of regular points (where $\op\w^kdF_i\neq 0$) is
open and dense. Let $M$ be an invariant submanifold through a
regular point $z\in Z_R\subset Z$. Then it is regular, i.e.,
$M\subset Z_R$. Let $M$ admit a regular open saturated
neighborhood $U_M$ (i.e., a fiber of $F$ through a point of $U_M$
belongs to $U_M$). For instance, any compact invariant submanifold
$M$ has such a neighborhood $U_M$. The restriction of functions
$F_i$ to $U_M$ defines a superintegrable system on $U_M$ which
obeys Definition \ref{i0}. In this case, one says that a
superintegrable system is considered around its invariant
submanifold $M$.
\end{rem}

Given a superintegrable system in accordance with Definition
\ref{i0}, the above mentioned generalization of the Mishchenko --
Fomenko theorem to noncompact invariant submanifolds states the
following \cite{fior2}.

\begin{theo} \label{nc0'} \mar{nc0'} Let the Hamiltonian vector fields $\vt_i$ of the
functions $F_i$ be complete, and let the fibers of the fibered
manifold $F$ (\ref{nc4}) be connected and mutually diffeomorphic.
Then the following hold.

(I) The fibers of $F$ (\ref{nc4}) are diffeomorphic to a toroidal
cylinder (\ref{g120}).

(II) Given a fiber $M$ of $F$ (\ref{nc4}), there exists its open
saturated neighborhood $U$ which is a trivial principal bundle
with the structure group (\ref{g120}).

(III) The neighborhood $U$ is provided with the bundle
(generalized action-angle) coordinates $(I_\la,p_s,q^s, y^\la)$,
$\la=1,\ldots, m$, $s=1,\ldots,n-m$, such that (i) the action
coordinates $(I_\la)$ are values of Casimir functions of the
coinduced Poisson structure $\{,\}_N$ on $F(U)$, (ii) the
generalized angle coordinates $(y^\la)$ are coordinates on a
toroidal cylinder, and (iii) the symplectic form $\Om$ on $U$
reads
\mar{cmp6}\beq
\Om= dI_\la\w dy^\la + dp_s\w dq^s. \label{cmp6}
\eeq
\end{theo}

\begin{proof}
It follows from item (iii) of Proposition \ref{nc8} that every
fiber $M$ of the fibered manifold (\ref{nc4}) is a maximal
integral manifold of the involutive distribution spanned by the
Hamiltonian vector fields $\up_\la$ of the pull-back $F^*C_\la$ of
$m$ independent Casimir functions $\{C_1,\ldots, C_m\}$ of the
Poisson structure $\{,\}_N$ (\ref{cmp1}) on an open neighborhood
$N_M$ of a point $F(M)\in N$. Let us put $U_M=F^{-1}(N_M)$. It is
an open saturated neighborhood of $M$. Consequently, invariant
submanifolds of a superintegrable system (\ref{i00}) on $U_M$ are
maximal integral manifolds of the partially integrable system
\mar{cmp4}\beq
C^*=(F^*C_1, \ldots, F^*C_m), \qquad 0<m\leq n, \label{cmp4}
\eeq
on a symplectic manifold $(U_M,\Om)$. Therefore, statements (I) --
(III) of Theorem \ref{nc0'} are the corollaries of forthcoming
Theorem \ref{nc6}. Its condition (i) is satisfied as follows. Let
$M'$ be an arbitrary fiber of the fibered manifold $F:U_M\to N_M$
(\ref{nc4}). Since
\mar{j21}\beq
F^*C_\la(z)= (C_\la\circ F)(z)= C_\la(F_i(z)),\qquad z\in M',
\label{j21}
\eeq
the Hamiltonian vector fields $\up_\la$ on $M'$ are $\Bbb
R$-linear combinations of Hamiltonian vector fields $\vt_i$ of the
functions $F_i$ It follows that $\up_\la$ are elements of a
finite-dimensional real Lie algebra of vector fields on $M'$
generated by the vector fields $\vt_i$. Since vector fields
$\vt_i$ are complete, the vector fields $v_\la$ on $M'$ also are
complete (see forthcoming Remark \ref{zz95}). Consequently, these
vector fields are complete on $U_M$ because they are vertical
vector fields on $U_M\to N$. The proof of Theorem \ref{nc6} shows
that the action coordinates $(I_\la)$ are values of Casimir
functions expressed into the original ones $C_\la$.
\end{proof}

\begin{rem} \label{zz95} \mar{zz95} If complete vector fields on a
smooth manifold constitute a basis for a finite-dimensional real
Lie algebra, any element of this Lie algebra is complete
\cite{palais}.
\end{rem}

\begin{rem} Since an open neighborhood $U$ in item (II) of
Theorem \ref{nc0'} is not contractible, unless $r=0$, the
generalized action-angle coordinates on $U$ sometimes are called
semi-local.
\end{rem}

\begin{rem} \label{zz90} \mar{zz90} The condition of the completeness of Hamiltonian
vector fields of the generating functions $F_i$ in Theorem
\ref{nc0'} is rather restrictive (see the Kepler system in Section
7). One can replace this condition with that the Hamiltonian
vector fields of the pull-back onto $Z$ of Casimir functions on
$N$ are complete.
\end{rem}

If the conditions of Theorem \ref{nc0'} are replaced with that the
fibers of the fibered manifold $F$ (\ref{nc4}) are compact and
connected, this theorem restarts the Mishchenko -- Fomenko one as
follows.

\begin{theo} \label{nc0} \mar{nc0}
Let the fibers of the fibered manifold $F$ (\ref{nc4}) be
connected and compact. Then they are diffeomorphic to a torus
$T^m$, and statements (II) -- (III) of Theorem \ref{nc0'} hold.
\end{theo}

\begin{rem}
In Theorem \ref{nc0}, the Hamiltonian vector fields $\up_\la$ are
complete because fibers of the fibered manifold $F$ (\ref{nc4})
are compact. As well known, any vector field on a compact manifold
is complete.
\end{rem}

If $F$ (\ref{i00}) is a completely integrable system, the
coinduced Poisson structure on $N$ equals zero, and the generating
functions $F_i$ are the pull-back of $n$ independent functions on
$N$. Then Theorems \ref{nc0} and \ref{nc0'} come to the Liouville
-- Arnold theorem \cite{arn1,laz} and its generalization (Theorem
\ref{cmp20}) to the case of noncompact invariant submanifolds
\cite{fior,book05}, respectively. In this case, the partially
integrable system $C^*$ (\ref{cmp4}) is exactly the original
completely integrable system $F$.

\begin{theo} \label{cmp20} \mar{cmp20} Given a completely integrable system
$F$ in accordance with Definition \ref{cmp21}, let the Hamiltonian
vector fields $\vt_i$ of the functions $F_i$ be complete, and let
the fibers of the fibered manifold $F$ (\ref{nc4}) be connected
and mutually diffeomorphic. Then items (I) and (II) of Theorem
\ref{nc0'} hold, and its item (III) is replaced with the following
one.

(III') The neighborhood $U$ is provided with the bundle
(generalized action-angle) coordinates $(I_\la,y^\la)$,
$\la=1,\ldots, n$, such that the angle coordinates $(y^\la)$ are
coordinates on a toroidal cylinder, and the symplectic form $\Om$
on $U$ reads
\mar{cmp66}\beq
\Om= dI_\la\w dy^\la. \label{cmp66}
\eeq
\end{theo}

Turn now to above mentioned Theorem \ref{nc6}. Recall that a
collection $\{S_1,\ldots, S_m\}$ of $m\leq n$ independent smooth
real functions in involution on a symplectic manifold $(Z,\Om)$ is
called a partially integrable system. Let us consider the map
\mar{g106}\beq
S:Z\to W\subset\Bbb R^m. \label{g106}
\eeq
Since functions $S_\la$ are everywhere independent, this map is a
submersion onto a domain $W\subset \Bbb R^m$, i.e., $S$
(\ref{g106}) is a fibered manifold of fiber dimension $2n-m$.
Hamiltonian vector fields $v_\la$ of functions $S_\la$ are
mutually commutative and independent. Consequently, they span an
$m$-dimensional involutive distribution on $Z$ whose maximal
integral manifolds constitute an isotropic foliation $\cF$ of $Z$.
Because functions $S_\la$ are constant on leaves of this
foliation, each fiber of a fibered manifold $Z\to W$ (\ref{g106})
is foliated by the leaves of the foliation $\cF$. If $m=n$, we are
in the case of a completely integrable system, and leaves of $\cF$
are connected components of fibers of the fibered manifold
(\ref{g106}). The Poincar\'e -- Lyapounov -- Nekhoroshev theorem
\cite{gaeta,nekh94} generalizes the Liouville -- Arnold one to a
partially integrable system if leaves of the foliation $\cF$ are
compact. It imposes a sufficient condition which Hamiltonian
vector fields $v_\la$ must satisfy in order that the foliation
$\cF$ is a fibered manifold \cite{gaeta,gaeta03}. Extending the
Poincar\'e -- Lyapounov -- Nekhoroshev theorem to the case of
noncompact invariant submanifolds, we in fact assume from the
beginning that these submanifolds form a fibered manifold
\cite{jmp03,book05}.

\begin{theo} \label{nc6} \mar{nc6}
Let a partially integrable system $\{S_1,\ldots,S_m\}$ on a
symplectic manifold $(Z,\Om)$ satisfy the following conditions.

(i) The Hamiltonian vector fields $v_\la$ of $S_\la$ are complete.

(ii) The foliation $\cF$ is a fibered manifold
\mar{d20'}\beq
\pi:Z\to N \label{d20'}
\eeq
whose fibers are mutually diffeomorphic.

\noindent Then the following hold.

(I) The fibers of $\cF$ are diffeomorphic to a toroidal cylinder
(\ref{g120}).

(II) Given a fiber $M$ of $\cF$, there exists its open saturated
neighborhood $U$ which is a trivial principal bundle with the
structure group (\ref{g120}).

(III) The neighborhood $U$ is provided with the bundle
(generalized action-angle) coordinates
\be
(I_\la,p_s,q^s,y^\la)\to (I_\la,p_s,q^s), \qquad \la=1,\ldots,m,
\quad s=1,\ldots n-m,
\ee
such that: (i) the action coordinates $(I_\la)$ (\ref{cmp25}) are
expressed into the values of the functions $(S_\la)$, (ii) the
angle coordinates $(y^\la)$ (\ref{cmp25}) are coordinates on a
toroidal cylinder, and (iii) the symplectic form $\Om$ on $U$
reads
\mar{cmp6'}\beq
\Om= dI_\la\w dy^\la + dp_s\w dq^s. \label{cmp6'}
\eeq
\end{theo}

\begin{proof} See Section 3 for the proof.
\end{proof}

If one supposes from the beginning that leaves of the foliation
$\cF$ are compact, the conditions of Theorem \ref{nc6} can be
replaced with that $\cF$ is a fibered manifold due to the
following.

\begin{prop} \label{cmp15} \mar{cmp15}
Any fibered manifold whose fibers are diffeomorphic either to
$\Bbb R^r$ or a connected compact manifold $K$ is a fiber bundle
\cite{meig}.
\end{prop}

\section{Proof of Theorem \ref{nc6}}

(I) In accordance with the well-known theorem \cite{onish,palais},
complete Hamiltonian vector fields $v_\la$ define an action of a
simply connected Lie group $G$ on $Z$. Because vector fields
$v_\la$ are mutually commutative, it is the additive group $\Bbb
R^m$ whose group space is coordinated by parameters $s^\la$ of the
flows with respect to the basis $\{e_\la=v_\la\}$ for its Lie
algebra. The orbits of the group $\Bbb R^m$ in $Z$ coincide with
the fibers of a fibered manifold $\cF$ (\ref{d20'}). Since vector
fields $v_\la$ are independent everywhere on $U$, the action of
$\Bbb R^m$ in $Z$ is locally free, i.e., isotropy groups of points
of $Z$ are discrete subgroups of the group $\Bbb R^m$. Given a
point $x\in \pi(U)$, the action of $\Bbb R^m$ in the fiber
$M_x=\pi^{-1}(x)$ factorizes as
\mar{d4}\beq
\Bbb R^m\times M_x\to G_x\times M_x\to M_x \label{d4}
\eeq
through the free transitive action in $M_x$ of the factor group
$G_x=\Bbb R^m/K_x$, where $K_x$ is the isotropy group of an
arbitrary point of $M_x$. It is the same group for all points of
$M_x$ because $\Bbb R^m$ is a commutative group. Clearly, $M_x$ is
diffeomorphic to the group space of $G_x$. Since the fibers $M_x$
are mutually diffeomorphic, all isotropy groups $K_x$ are
isomorphic to the group $\Bbb Z_r$ for some fixed $0\leq r\leq m$.
Accordingly, the groups $G_x$ are isomorphic to the additive group
$\Bbb R^{m-r}\times T^r$. This proves statement (I) of Theorem
\ref{nc6}.

(II) Because $\cF$ is a fibered manifold, one can always choose an
open fibered neighborhood $U$ of its fiber $M$ such that $\pi(U)$
is a domain and a fibered manifold
\mar{d20}\beq
\pi:U\to \pi(U)\subset N \label{d20}
\eeq
admits a section $\si$. Let us bring the fibered manifold
(\ref{d20}) into a principal bundle with the structure group
$G_0$, where we denote $\{0\}=\pi(M)$. For this purpose, let us
determine isomorphisms $\rho_x: G_0\to G_x$ of the group $G_0$ to
the groups $G_x$, $x\in \pi(U)$. Then a desired fiberwise action
of $G_0$ in $U$ is defined by the law
\mar{d5}\beq
G_0\times M_x\to\rho_x(G_0)\times M_x\to M_x. \label{d5}
\eeq
Generators of each isotropy subgroup $K_x$ of $\Bbb R^m$ are given
by $r$ linearly independent vectors of the group space $\Bbb R^m$.
One can show that there exist ordered collections of generators
$(v_1(x),\ldots,v_r(x))$ of the groups $K_x$ such that $x\mapsto
v_i(x)$ are smooth $\Bbb R^m$-valued fields on $\pi(U)$. Indeed,
given a vector $v_i(0)$ and a section $\si$ of the fibered
manifold (\ref{d20}), each field $v_i(x)=(s^\al(x))$ is the unique
smooth solution of the equation
\be
g(s^\al)\si(x)=\si(x), \qquad  (s^\al(0))=v_i(0),
\ee
on an open neighborhood of $\{0\}$. Let us consider the
decomposition
\be
v_i(0)=B_i^a(0) e_a + C_i^j(0) e_j, \qquad a=1,\ldots,m-r, \qquad
j=1,\ldots, r,
\ee
where $C_i^j(0)$ is a non-degenerate matrix. Since the fields
$v_i(x)$ are smooth, there exists an open neighborhood of $\{0\}$,
say $\pi(U)$ again, where the matrices $C_i^j(x)$ are
non-degenerate. Then
\mar{d6}\beq
A_x=\left(
\begin{array}{ccc}
\id & \qquad & (B(x)-B(0))C^{-1}(0) \\
0 & & C(x)C^{-1}(0)
\end{array}
\right) \label{d6}
\eeq
is a unique linear morphism of the vector space $\Bbb R^m$ which
transforms the frame $v_\la(0)=\{e_a,v_i(0)\}$ into the frame
$v_\la(x)=\{e_a,v_i(x)\}$. Since it also is an automorphism of the
group $\Bbb R^m$ sending $K_0$ onto $K_x$, we obtain a desired
isomorphism $\rho_x$ of the group $G_0$ to the group $G_x$. Let an
element $g$ of the group $G_0$ be the coset of an element
$g(s^\la)$ of the group $\Bbb R^m$. Then it acts in $M_x$ by the
rule (\ref{d5}) just as the element $g((A_x^{-1})^\la_\bt s^\bt)$
of the group $\Bbb R^m$ does. Since entries of the matrix $A_x$
(\ref{d6}) are smooth functions on $\pi(U)$, this action of the
group $G_0$ in $U$ is smooth. It is free, and $U/G_0=\pi(U)$. Then
the fibered manifold $U\to \pi(U)$ is a trivial principal bundle
with the structure group $G_0$.

(III) Given a section $\si$ of the principal bundle $U\to \pi(U)$,
its trivialization $U=\pi(U)\times G_0$ is defined by assigning
the points $\rho^{-1}(g_x)$ of the group space $G_0$ to the points
$g_x\si(x)$, $g_x\in G_x$, of a fiber $M_x$. Let us endow $G_0$
with the standard coordinate atlas $(r^\la)=(t^a,\vf^i)$ of the
group $\Bbb R^{m-r}\times T^r$. Then we provide $U$ with the
trivialization
\mar{z10}\beq
U=\pi(U)\times(\Bbb R^{m-r}\times T^r)\to\pi(U) \label{z10}
\eeq
with respect to the fiber coordinates $(t^a,\vf^i)$. The vector
fields $v_\la$ on $U$ relative to these coordinates read
\mar{ww25}\beq
v_a=\dr_a, \qquad v_i=-(BC^{-1})^a_i(x)\dr_a +
(C^{-1})_i^k(x)\dr_k.\label{ww25}
\eeq

In order to specify coordinates on the base $\pi(U)$ of the
trivial bundle (\ref{z10}), let us consider the fibered manifold
$S$ (\ref{g106}). It factorizes as
\be
S: U\ar^\pi \pi(U)\ar^{\pi'} S(U), \qquad \pi'=S\circ\si,
\ee
through the fiber bundle $\pi$. The map $\pi'$ also is a fibered
manifold. One can always restrict the domain $\pi(U)$ to a chart
of the fibered manifold $\pi'$, say $\pi(U)$ again. Then
$\pi(U)\to S(U)$ is a trivial bundle $\pi(U)=S(U)\times V$, and so
is $U\to S(U)$. Thus, we have the composite bundle
\mar{z10'}\beq
U=S(U)\times V\times (\Bbb R^{m-r}\times T^r)\to S(U)\times V\to
S(U). \label{z10'}
\eeq
Let us provide its base $S(U)$ with the coordinates $(J_\la)$ such
that
\mar{cmp23}\beq
J_\la\circ S=S_\la. \label{cmp23}
\eeq
Then $\pi(U)$ can be equipped with the bundle coordinates $(J_\la,
x^A)$, $A=1,\ldots, 2(n-m)$, and $(J_\la, x^A, t^a,\vf^i)$ are
coordinates on $U$ (\ref{z10'}). Since fibers of $U\to \pi(U)$ are
isotropic, a symplectic form $\Om$ on $U$ relative to the
coordinates $(J_\la,x^A,r^\la)$ reads
\mar{d23}\beq
\Om=\Om^{\al\bt}dJ_\al\w dJ_\bt + \Om^\al_\bt dJ_\al\w dr^\bt +
\Om_{AB}dx^A\w dx^B +\Om_A^\la dJ_\la\w dx^A +
  \Om_{A\bt} dx^A\w dr^\bt. \label{d23}
\eeq
The Hamiltonian vector fields $v_\la=v_\la^\m\dr_\m$ (\ref{ww25})
obey the relations $v_\la\rfloor\Om=-dJ_\la$ which result in the
coordinate conditions
\mar{ww22}\beq
  \Om^\al_\bt v^\bt_\la=\dl^\al_\la, \qquad \Om_{A\bt}v^\bt_\la=0.
\label{ww22}
\eeq
The first of them shows that $\Om^\al_\bt$ is a non-degenerate
matrix independent of coordinates $y^\la$. Then the second one
implies that $\Om_{A\bt}=0$.

By virtue of the well-known K\"unneth formula for the de Rham
cohomology of manifold products, the closed form $\Om$ (\ref{d23})
is exact, i.e., $\Om=d\Xi$ where the Liouville form $\Xi$ is
\be
\Xi=\Xi^\al(J_\la,x^B,r^\la)dJ_\al + \Xi_i(J_\la,x^B) d\vf^i
+\Xi_A(J_\la,x^B,r^\la)dx^A.
\ee
Since $\Xi_a=0$ and $\Xi_i$ are independent of $\vf^i$, it follows
from the relations
\be
\Om_{A\bt}=\dr_A\Xi_\bt-\dr_\bt\Xi_A=0
\ee
that $\Xi_A$ are independent of coordinates $t^a$ and are at most
affine in $\vf^i$. Since $\vf^i$ are cyclic coordinates, $\Xi_A$
are independent of $\vf^i$. Hence, $\Xi_i$ are  independent of
coordinates $x^A$, and the Liouville form reads
\mar{ac2}\beq
\Xi=\Xi^\al(J_\la,x^B,r^\la)dJ_\al + \Xi_i(J_\la) d\vf^i
+\Xi_A(J_\la,x^B)dx^A. \label{ac2}
\eeq
Because entries $\Om^\al_\bt$ of $d\Xi=\Om$ are independent of
$r^\la$, we obtain the following.

(i) $\Om^\la_i=\dr^\la\Xi_i-\dr_i\Xi^\la$. Consequently,
$\dr_i\Xi^\la$ are independent of $\vf^i$, and so are $\Xi^\la$
since $\vf^i$ are cyclic coordinates. Hence,
$\Om^\la_i=\dr^\la\Xi_i$ and $\dr_i\rfloor\Om=-d\Xi_i$. A glance
at the last equality shows that $\dr_i$ are Hamiltonian vector
fields. It follows that, from the beginning, one can separate $r$
generating functions on $U$, say $S_i$ again, whose Hamiltonian
vector fields are tangent to invariant tori. In this case, the
matrix $B$ in the expressions (\ref{d6}) and (\ref{ww25})
vanishes, and the Hamiltonian vector fields $v_\la$ (\ref{ww25})
read
\mar{ww25'}\beq
v_a=\dr_a, \qquad v_i=(C^{-1})_i^k\dr_k. \label{ww25'}
\eeq
Moreover, the coordinates $t^a$ are exactly the flow parameters
$s^a$. Substituting the expressions (\ref{ww25'}) into the first
condition (\ref{ww22}), we obtain
\be
\Om=\Om^{\al\bt}dJ_\al\w dJ_\bt +dJ_a\w ds^a + C^i_k dJ_i\w d\vf^k
+ \Om_{AB}dx^A\w dx^B +\Om_A^\la dJ_\la\w dx^A.
\ee
It follows that $\Xi_i$ are independent of $J_a$, and so are
$C^k_i=\dr^k\Xi_i$.

(ii) $\Om^\la_a=-\dr_a\Xi^\la=\dl^\la_a$. Hence,
$\Xi^a=-s^a+E^a(J_\la)$ and $\Xi^i=E^i(J_\la,x^B)$ are independent
of $s^a$.

In view of items (i) -- (ii), the Liouville form $\Xi$ (\ref{ac2})
reads
\be
\Xi=(-s^a+E^a(J_\la,x^B))dJ_a + E^i(J_\la,x^B) dJ_i + \Xi_i(J_j)
d\vf^i + \Xi_A(J_\la,x^B)dx^A.
\ee
Since the matrix $\dr^k\Xi_i$ is non-degenerate, we can perform
the coordinate transformations
\mar{cmp25}\beq
I_a=J_a, \quad I_i=\Xi_i(J_j), \quad r'^a=-s^a+E^a(J_\la,x^B),
\quad r'^i=\vf^i-E^j(J_\la,x^B)\frac{\dr J_j}{\dr I_i}.
\label{cmp25}
\eeq
These transformations bring $\Om$ into the form
\mar{d26}\beq
\Om= dI_\la\w d r'^\la +\Om_{AB}(I_\m,x^C) dx^A\w dx^B +
\Om_A^\la(I_\m,x^C) dI_\la\w dx^A. \label{d26}
\eeq
Since functions $I_\la$ are in involution and their Hamiltonian
vector fields $\dr_\la$ mutually commute, a point $z\in M$ has an
open neigbourhood $U_z=\pi(U_z)\times O_z$, $O_z\subset \Bbb
R^{m-r}\times T^r,$ endowed with local Darboux coordinates
$(I_\la,p_s,q^s, y^\la)$, $s=1,\ldots,n-m$, such that the
symplectic form $\Om$ (\ref{d26}) is given by the expression
\mar{dd12}\beq
\Om= dI_\la\w d y^\la + dp_s\w dq^s. \label{dd12}
\eeq
Here, $y^\la(I_\la,x^A,r'^\al)$ are local functions
\mar{dd11}\beq
y^\la=r'^\la + f^\la(I_\la,x^A) \label{dd11}
\eeq
on $U_z$. With the group $G$, one can extend these functions to
the open neighborhood
\be
\pi(U_z)\times \Bbb R^{k-m}\times T^m
\ee
of $M$, say $U$ again, by the law
\be
y^\la(I_\la,x^A,G(z)^\al)= G(z)^\la + f^\la(I_\la,x^A).
\ee
Substituting the functions (\ref{dd11}) on $U$ into the expression
(\ref{d26}), one brings the symplectic form $\Om$ into the
canonical form (\ref{dd12}) on $U$.

\section{Globally superintegrable systems}

To study a superintegrable system, one conventionally considers it
with respect to generalized action-angle coordinates. A problem is
that, restricted to an action-angle coordinate chart on an open
subbundle $U$ of the fibered manifold $Z\to N$ (\ref{nc4}), a
superintegrable system becomes different from the original one
since there is no morphism of the Poisson algebra $C^\infty(U)$ on
$(U,\Om)$ to that $C^\infty(Z)$ on $(Z,\Om)$. Moreover, a
superitegrable system on $U$ need not satisfy the conditions of
Theorem \ref{nc0'} because it may happen that the Hamiltonian
vector fields of the generating functions on $U$ are not complete.
To describe superintegrable systems in terms of generalized
action-angle coordinates, we therefore follow the notion of a
globally superintegrable system.

\begin{defi} \mar{cmp30} \label{cmp30} A superintegrable systems
$F$ (\ref{i00}) on a symplectic manifold $(Z,\Om)$ by Definition
\ref{i0} is called globally superintegrable if there exist global
generalized action-angle coordinates
\mar{cmp31}\beq
(I_\la, x^A, y^\la), \qquad \la=1,\ldots,m, \qquad A=1,\ldots,
2(n-m), \label{cmp31}
\eeq
such that: (i) the action coordinates $(I_\la)$ are expressed into
the values of some Casimir functions $C_\la$ on the Poisson
manifold $(N,\{,\}_N)$, (ii) the angle coordinates $(y^\la)$ are
coordinates on the toroidal cylinder (\ref{g120}), and (iii) the
symplectic form $\Om$ on $Z$ reads
\mar{cmp32}\beq
\Om= dI_\la\w d y^\la +\Om_{AB}(I_\m,x^C) dx^A\w dx^B.
\label{cmp32}
\eeq
\end{defi}

It is readily observed that the semi-local generalized
action-angle coordinates on $U$ in Theorem \ref{nc0'} are global
in accordance with Definition \ref{cmp30}.

Forthcoming Theorem \ref{cmp34} provides the sufficient conditions
of the existence of global generalized action-angle coordinates of
a superintegrable system on a symplectic manifold $(Z,\Om)$
\cite{jmp07}.

\begin{theo} \label{cmp34} \mar{cmp34} A superintegrable system $F$ on
a symplectic manifold $(Z,\Om)$ is globally superintegrable if the
following conditions hold.

(i) Hamiltonian vector fields $\vt_i$ of the generating functions
$F_i$ are complete.

(ii) The fibered manifold $F$ (\ref{nc4}) is a fiber bundle with
connected fibers.

(iii) Its base $N$ is simply connected and the cohomology
$H^2(N,\Bbb Z)$ is trivial

(iv) The coinduced Poisson structure $\{,\}_N$ on a base $N$
admits $m$ independent Casimir functions $C_\la$.
\end{theo}

\begin{proof} Theorem \ref{cmp34} is a corollary of
Theorem \ref{cmp35} below which is a global generalization of
Theorem \ref{nc6}. In accordance with Theorem \ref{cmp35}, we have
a composite fibered manifold
\mar{g150}\beq
Z\ar^F N\ar^C W, \label{g150}
\eeq
where $C:N\to W$ is a fibered manifold of level surfaces of the
Casimir functions $C_\la$ (which coincides with the symplectic
foliation of a Poisson manifold $N$). The composite fibered
manifold (\ref{g150}) is provided with the adapted fibered
coordinates $(J_\la, x^A, r^\la)$ (\ref{g108}), where $J_\la$ are
values of independent Casimir functions and $(r^\la)=(t^a,\vf^i)$
are coordinates on a toroidal cylinder. Since $C_\la=J_\la$ are
Casimir functions on $N$, the symplectic form $\Om$ (\ref{g103})
on $Z$ reads
\mar{g141}\beq
\Om=\Om^\al_\bt dJ_\al\w r^\bt + \Om_{\al A}dr^\al\w dx^A +
\Om_{AB} dx^A\w dx^B. \label{g141}
\eeq
In particular, it follows that transition functions of coordinates
$x^A$ on $N$ are independent of coordinates $J_\la$, i.e., $C:V\to
W$ is a trivial bundle. By virtue of Lemma \ref{g144} below, the
symplectic form (\ref{g141}) is exact, i.e., $\Om=d\Xi$, where the
Liouville form $\Xi$ (\ref{g113}) is
\be
\Xi=\Xi^\la(J_\al,r^\al)dJ_\la + \Xi_i(J_\al) d\vf^i
+\Xi_A(x^B)dx^A.
\ee
Then the coordinate transformations (\ref{g142}):
\mar{g151}\beq
I_a=J_a, \quad I_i=\Xi_i(J_j), \quad y^a = -\Xi^a=t^a-E^a(J_\la),
\quad y^i =\vf^i-\Xi^j(J_\la)\frac{\dr J_j}{\dr I_i}, \label{g151}
\eeq
bring $\Om$ (\ref{g141}) into the form (\ref{cmp32}). In
comparison with the general case (\ref{g142}), the coordinate
transformations (\ref{g151}) are independent of coordinates $x^A$.
Therefore, the angle coordinates $y^i$ possess identity transition
functions on $N$.
\end{proof}

Theorem \ref{cmp34} restarts Theorem \ref{nc0'} if one considers
an open subset $V$ of $N$ admitting the Darboux coordinates $x^A$
on the symplectic leaves of $U$.

Note that, if invariant submanifolds of a superintegrable system
are assumed to be connected and compact, condition (i) of Theorem
\ref{cmp34} is unnecessary since vector fields $v_\la$ on compact
fibers of $F$ are complete. Condition (ii) also holds by virtue of
Proposition \ref{cmp15}. In this case, Theorem \ref{cmp34}
reproduces the well known result in \cite{daz}.

If  $F$ in Theorem \ref{cmp34} is a completely integrable system,
the coinduced Poisson structure on $N$ equals zero, the generating
functions $F_i$ are the pull-back of $n$ independent functions on
$N$, and Theorem \ref{cmp34} coincides with Theorem 4 in
\cite{jmp07}.

Turn now to the above mentioned Theorem \ref{cmp35}.

\begin{theo} \label{cmp35} \mar{cmp35}
Let a partially integrable system $\{S_1,\ldots,S_m\}$ on a
symplectic manifold $(Z,\Om)$ satisfy the following conditions.

(i) The Hamiltonian vector fields $v_\la$ of $S_\la$ are complete.

(ii) The foliation $\cF$ is a fiber bundle
\mar{cmp40}\beq
\pi:Z\to N. \label{cmp40}
\eeq

(iii) Its base $N$ is simply connected and the cohomology
$H^2(N,\Bbb Z)$ is trivial.

\noindent Then the following hold.

(I) The fiber bundle $\cF$ is a trivial principal bundle with the
structure group (\ref{g120}), and we have a composite fibered
manifold
\mar{g107}\beq
S=\zeta\circ\pi: Z\ar N\ar W, \label{g107}
\eeq
where $N\to W$ however need not be a fiber bundle.

(II) The fibered manifold (\ref{g107}) is provided with the
fibered generalized action-angle coordinates
\be
(I_\la,x^A,y^\la)\to (I_\la,x^A)\to (I_\la), \qquad
\la=1,\ldots,m, \quad A=1,\ldots, 2(n-m),
\ee
such that: (i) the action coordinates $(I_\la)$ (\ref{g142}) are
expressed into the values of the functions $(S_\la)$ and they
possess identity transition functions, (ii) the angle coordinates
$(y^\la)$ (\ref{g142}) are coordinates on a toroidal cylinder,
(iii) the symplectic form $\Om$ on $U$ reads
\mar{nc3'}\beq
\Om= dI_\la\w dy^\la + \Om_A^\la dI_\la\w dx^A+ \Om_{AB} dx^A\w
dx^B. \label{nc3'}
\eeq
\end{theo}

\begin{proof}
See Section 5 for the proof.
\end{proof}

It follows from the proof of Theorem \ref{cmp35} that its
condition (iii) and, accordingly, condition (iii) of Theorem
\ref{cmp34} guarantee that fiber bundles $F$ in conditions (ii) of
these theorems are trivial. Therefore, Theorem \ref{cmp34} can be
reformulated as follows.

\begin{theo} \mar{cmp36} \label{cmp36} A superintegrable system $F$ on
a symplectic manifold $(Z,\Om)$ is globally superintegrable iff
the following conditions hold.

(i) The fibered manifold $F$ (\ref{nc4}) is a trivial fiber
bundle.

(ii) The coinduced Poisson structure $\{,\}_N$ on a base $N$
admits $m$ independent Casimir functions $C_\la$ such that
Hamiltonian vector fields of their pull-back $F^*C_\la$ are
complete.
\end{theo}

\begin{rem} \label{zz96} \mar{zz96}
It follows from Remark \ref{zz95} and condition (ii) of Theorem
\ref{cmp36} that a Hamiltonian vector field of the the pull-back
$F^*C$ of any Casimir function $C$ on a Poisson manifold $N$ is
complete.
\end{rem}

\section{Proof of Theorem \ref{cmp35}}

Following part (I) of the proof of Theorem \ref{nc6}, one can show
that a typical fiber of the fiber bundle (\ref{cmp40}) is the
toroidal cylinder (\ref{g120}).

Let us bring the fiber bundle (\ref{cmp40}) into a principal
bundle with the structure group (\ref{g120}). Generators of each
isotropy subgroup $K_x$ of $\Bbb R^m$ are given by $r$ linearly
independent vectors $u_i(x)$ of the group space $\Bbb R^m$. These
vectors are assembled into an $r$-fold covering $K\to N$. This is
a subbundle of the trivial bundle
\mar{g101}\beq
N\times R^m\to N \label{g101}
\eeq
whose local sections are local smooth sections of the fiber bundle
(\ref{g101}). Such a section over an open neighborhood of a point
$x\in N$ is given by a unique local solution $s^\la(x')e_\la$ of
the equation
\be
g(s^\la)\si(x')=\exp(s^\la v_\la)\si(x')=\si(x'), \qquad
s^\la(x)e_\la=u_i(x),
\ee
where $\si$ is an arbitrary local section of the fiber bundle
$Z\to N$ over an open neighborhood of $x$. Since $N$ is simply
connected, the covering $K\to N$ admits $r$ everywhere different
global sections $u_i$ which are global smooth sections
$u_i(x)=u^\la_i(x)e_\la$ of the fiber bundle (\ref{g101}). Let us
fix a point of $N$ further denoted by $\{0\}$. One can determine
linear combinations of the functions $S_\la$, say again $S_\la$,
such that $u_i(0)=e_i$, $i=m-r,\ldots,m$, and the group $G_0$ is
identified to the group $\Bbb R^{m-r}\times T^r$. Let $E_x$ denote
the $r$-dimensional subspace of $\Bbb R^m$ passing through the
points $u_1(x),\ldots,u_r(x)$. The spaces $E_x$, $x\in N$,
constitute an $r$-dimensional subbundle $E\to N$ of the trivial
bundle (\ref{g101}). Moreover, the latter is split into the
Whitney sum of vector bundles $E\oplus E'$, where $E'_x=\Bbb
R^m/E_x$ \cite{hir}. Then there is a global smooth section $\g$ of
the trivial principal bundle $N\times GL(m,\Bbb R)\to N$ such that
$\g(x)$ is a morphism of $E_0$ onto $E_x$, where
$u_i(x)=\g(x)(e_i)=\g_i^\la e_\la$. This morphism also is an
automorphism of the group $\Bbb R^m$ sending $K_0$ onto $K_x$.
Therefore, it provides a group isomorphism $\rho_x: G_0\to G_x$.
With these isomorphisms, one can define the fiberwise action of
the group $G_0$ on $Z$ given by the law
\mar{d5'}\beq
G_0\times M_x\to\rho_x(G_0)\times M_x\to M_x. \label{d5'}
\eeq
Namely, let an element of the group $G_0$ be the coset
$g(s^\la)/K_0$ of an element $g(s^\la)$ of the group $\Bbb R^m$.
Then it acts on $M_x$ by the rule (\ref{d5'}) just as the coset
$g((\g(x)^{-1})^\la_\bt s^\bt)/K_x$ of an element
$g((\g(x)^{-1})^\la_\bt s^\bt)$ of $\Bbb R^m$ does. Since entries
of the matrix $\g$ are smooth functions on $N$, the action
(\ref{d5'}) of the group $G_0$ on $Z$ is smooth. It is free, and
$Z/G_0=N$. Thus, $Z\to N$ (\ref{cmp40}) is a principal bundle with
the structure group $G_0=\Bbb R^{m-r}\times T^r$.

Furthermore, this principal bundle over a paracompact smooth
manifold $N$ is trivial as follows. In accordance with the
well-known theorem \cite{hir}, its structure group $G_0$
(\ref{g120}) is reducible to the maximal compact subgroup $T^r$,
which also is the maximal compact subgroup of the group product
$\op\times^rGL(1,\Bbb C)$. Therefore, the equivalence classes of
$T^r$-principal bundles $\xi$ are defined as
\be
c(\xi)=c(\xi_1\oplus\cdots\oplus \xi_r)=(1+c_1(\xi_1))\cdots
(1+c_1(\xi_r))
\ee
by the Chern classes $c_1(\xi_i)\in H^2(N,\Bbb Z)$ of
$U(1)$-principal bundles $\xi_i$ over $N$ \cite{hir}. Since the
cohomology group $H^2(N,\Bbb Z)$ of $N$ is trivial, all Chern
classes $c_1$ are trivial, and the principal bundle $Z\to N$ over
a contractible base also is trivial. This principal bundle can be
provided with the following coordinate atlas.

Let us consider the fibered manifold $S:Z\to W$ (\ref{g106}).
Because functions $S_\la$ are constant on fibers of the fiber
bundle $Z\to N$ (\ref{cmp40}), the fibered manifold (\ref{g106})
factorizes through the fiber bundle (\ref{cmp40}), and we have the
composite fibered manifold (\ref{g107}). Let us provide the
principal bundle $Z\to N$ with a trivialization
\mar{g110}\beq
Z=N\times \Bbb R^{m-r}\times T^r\to N, \label{g110}
\eeq
whose fibers are endowed with the standard coordinates
$(r^\la)=(t^a,\vf^i)$ on the toroidal cylinder (\ref{g120}). Then
the composite fibered manifold (\ref{g107}) is provided with the
fibered coordinates
\mar{g108}\ben
&& (J_\la,x^A,t^a,\vf^i), \label{g108}\\
&& \la=1,\ldots, m, \quad A=1, \ldots, 2(n-m), \quad a=1, \ldots,
m-r, \quad i=1,\ldots, r, \nonumber
\een
where $J_\la$ (\ref{cmp23}) are coordinates on the base $W$
induced by Cartesian coordinates on $\Bbb R^m$, and $(J_\la, x^A)$
are fibered coordinates on the fibered manifold $\zeta:N\to W$.
The coordinates $J_\la$ on $W\subset \Bbb R^m$ and the coordinates
$(t^a,\vf^i)$ on the trivial bundle (\ref{g110}) possess the
identity transition functions, while the transition function of
coordinates $(x^A)$ depends on the coordinates $(J_\la)$ in
general.

The Hamiltonian vector fields $v_\la$ on $Z$ relative to the
coordinates (\ref{g108}) take the form
\mar{ww25a}\beq
v_\la=v_\la^a(x)\dr_a + v^i_\la(x)\dr_i. \label{ww25a}
\eeq
Since these vector fields commute (i.e., fibers of $Z\to N$ are
isotropic), the symplectic form $\Om$ on $Z$ reads
\mar{g103}\beq
\Om=\Om^\al_\bt dJ_\al\w dr^\bt + \Om_{\al A}dr^\al\w dx^A +
\Om^{\al\bt} dJ_\al\w dJ_\bt + \Om^\al_A d J_\al\w dx^A +\Om_{AB}
dx^A\w dx^B. \label{g103}
\eeq

\begin{lem} \label{g144} \mar{g144}
The symplectic form $\Om$ (\ref{g103}) is exact.
\end{lem}

\begin{proof} In accordance with the well-known K\"unneth formula,
the de Rham cohomology group of the product (\ref{g110}) reads
\be
H^2(Z)=H^2(N)\oplus H^1(N)\ot H^1(T^r) \oplus H^2(T^r).
\ee
By the de Rham theorem \cite{hir}, the de Rham cohomology $H^2(N)$
is isomorphic to the cohomology $H^2(N,\Bbb R)$ of $N$ with
coefficients in the constant sheaf $\Bbb R$. It is trivial since
$H^2(N,\Bbb R)=H^2(N,\Bbb Z)\ot\Bbb R$ where $H^2(N,\Bbb Z)$ is
trivial. The first cohomology group $H^1(N)$ of $N$ is trivial
because $N$ is simply connected. Consequently, $H^2(Z)=H^2(T^r)$.
Then the closed form $\Om$ (\ref{g103}) is exact since it does not
contain the term $\Om_{ij}d\vf^i\w d\vf^j$.
\end{proof}

Thus, we can write
\mar{g113}\beq
\Om=d\Xi, \qquad \Xi=\Xi^\la(J_\al,x^B,r^\al) dJ_\la +
\Xi_\la(J_\al,x^B) dr^\la +\Xi_A(J_\al,x^B,r^\al) dx^A.
\label{g113}
\eeq
Up to an exact summand, the Liouville form $\Xi$ (\ref{g113}) is
brought into the form
\be
\Xi=\Xi^\la(J_\al,x^B,r^\al) dJ_\la + \Xi_i(J_\al,x^B) d\vf^i
+\Xi_A(J_\al,x^B,r^\al) dx^A,
\ee
i.e., it does not contain the term $\Xi_a dt^a$.

The Hamiltonian vector fields $v_\la$ (\ref{ww25a}) obey the
relations $v_\la\rfloor\Om=-dJ_\la$, which result in the
coordinate conditions (\ref{ww22}). Then following the proof of
Theorem \ref{nc6}, we can show that a symplectic form $\Om$ on $Z$
is given by the expression (\ref{nc3'}) with respect to the
coordinates
\mar{g142}\beq
I_a=J_a, \quad I_i=\Xi_i(J_j), \quad y^a =
-\Xi^a=t^a-E^a(J_\la,x^B), \quad y^i
=\vf^i-\Xi^j(J_\la,x^B)\frac{\dr J_j}{\dr I_i}. \label{g142}
\eeq

\section{Superintegrable Hamiltonian systems}

In autonomous (symplectic) Hamiltonian mechanics, one considers
superintegrable systems whose generating functions are integrals
of motion, i.e., they are in involution with a Hamiltonian $\cH$,
and the functions $(\cH,F_1,\ldots,F_k)$ are nowhere independent,
i.e.,
\mar{cmp11,'}\ben
&&\{\cH, F_i\}=0, \label{cmp11}\\
&& dH\w(\op\w^kdF_i)=0. \label{cmp11'}
\een.

In order that an evolution of Hamiltonian system can be defined at
any instant $t\in\Bbb R$, one supposes that the Hamiltonian vector
field of its Hamiltonian is complete. By virtue of Remark
\ref{zz96} and forthcoming Proposition \ref{cmp12}, a Hamiltonian
of a superintegrable system always satisfies this condition.

\begin{prop} \label{cmp12} \mar{cmp12}
It follows from the equality (\ref{cmp11'}) that a Hamiltonian
$\cH$ is constant on the invariant submanifolds. Therefore, it is
the pull-back of a function on $N$ which is a Casimir function of
the Poisson structure (\ref{cmp1}) because of the conditions
(\ref{cmp11}).
\end{prop}

Proposition \ref{cmp12} leads to the following.

\begin{prop} \label{zz1} \mar{zz1}
Let $\cH$ be a Hamiltonian of a globally superintegrable system
provided with the generalized action-angle coordinates $(I_\la,
x^A, y^\la)$ (\ref{cmp31}). Then a Hamiltonian $\cH$ depends only
on the action coordinates $I_\la$. Consequently, the equations of
motion of a globally superintergable system take the form
\be \dot
y^\la=\frac{\dr \cH}{\dr I_\la}, \qquad I_\la={\rm const.}, \qquad
x^A={\rm const.}
\ee
\end{prop}

Following the original Mishchenko--Fomenko theorem, let us mention
superintegrable systems whose generating functions
$\{F_1,\ldots,F_k\}$ form a $k$-dimensional real Lie algebra $\cG$
of corank $m$ with the commutation relations
\mar{zz60}\beq
\{F_i,F_j\}= c_{ij}^h F_h, \qquad c_{ij}^h={\rm const.}
\label{zz60}
\eeq
Then $F$ (\ref{nc4}) is a momentum mapping of $Z$ to the Lie
coalgebra $\cG^*$ provided with the coordinates $x_i$ in item (i)
of Definition \ref{i0} \cite{book05,guil}. In this case, the
coinduced Poisson structure $\{,\}_N$ coincides with the canonical
Lie--Poisson structure on $\cG^*$ given by the Poisson bivector
field
\be
w=\frac12 c_{ij}^h x_h\dr^i\w\dr^j.
\ee
Let $V$ be an open subset of $\cG^*$ such that conditions (i) and
(ii) of Theorem \ref{cmp36} are satisfied. Then an open subset
$F^{-1}(V)\subset Z$ is provided with the generalized action-angle
coordinates.

\begin{rem}
Let Hamiltonian vector fields $\vt_i$ of the generating functions
$F_i$ which form a Lie algebra $\cG$ be complete. Then they define
a locally free Hamiltonian action on $Z$ of some simply connected
Lie group $G$ whose Lie algebra is isomorphic to $\cG$
\cite{onish,palais}. Orbits of $G$ coincide with $k$-dimensional
maximal integral manifolds of the regular distribution $\cV$ on
$Z$ spanned by Hamiltonian vector fields $\vt_i$ \cite{susm}.
Furthermore, Casimir functions of the Lie--Poisson structure on
$\cG^*$ are exactly the coadjoint invariant functions on $\cG^*$.
They are constant on orbits of the coadjoint action of $G$ on
$\cG^*$ which coincide with leaves of the symplectic foliation of
$\cG^*$.
\end{rem}

\begin{theo} \label{zz34} \mar{zz34} Let a globally superintegrable Hamiltonian system
on a symplectic manifold $Z$ obey the following conditions.

(i) It is maximally superintegrable.

(ii) Its Hamiltonian $H$ is regular, i.e, $dH$ nowhere vanishes.

(iii) Its generating functions $F_i$ constitute a finite
dimensional real Lie algebra and their Hamiltonian vector fields
are complete.

\noindent Then any integral of motion of this Hamiltonian system
is the pull-back of a function on a base $N$ of the fibration $F$
(\ref{nc4}). In other words, it is expressed into the integrals of
motion $F_i$.
\end{theo}

\begin{proof}
The proof is based on the following. A Hamiltonian vector field of
a function $f$ on $Z$ lives in the one-codimensional regular
distribution $\cV$ on $Z$ spanned by Hamiltonian vector fields
$\vt_i$ iff $f$ is the pull-back of a function on a base $N$ of
the fibration $F$ (\ref{nc4}). A Hamiltonian $H$ brings $Z$ into a
fibered manifold of its level surfaces whose vertical tangent
bundle coincide with $\cV$. Therefore, a Hamiltonian vector field
of any integral of motion of $H$ lives in $\cV$.
\end{proof}

It may happen that, given a Hamiltonian $\cH$ of a Hamiltonian
system on a symplectic manifold $Z$, we have different
superintegrable Hamiltonian systems on different open subsets of
$Z$. For instance, this is the case of the Kepler system.

\section{Kepler system}

We consider the Kepler system on a plane $\Bbb R^2$. Its phase
space is $T^*\Bbb R^2=\Bbb R^4$ provided with the Cartesian
coordinates $(q_i,p_i)$, $i=1,2$, and the canonical symplectic
form
\mar{zz43}\beq
\Om=\op\sum_idp_i\w dq_i. \label{zz43}
\eeq
Let us denote
\be
p=(\op\sum_i(p_i)^2)^{1/2}, \qquad r=(\op\sum_i(q^i)^2)^{1/2},
\qquad (p,q)=\op\sum_ip_iq_i.
\ee
A Hamiltonian of the Kepler system reads
\mar{zz41}\beq
H=\frac12p^2-\frac1r. \label{zz41}
\eeq
The Kepler system is a Hamiltonian system on a symplectic manifold
\mar{zz42}\beq
Z=\Bbb R^4\setminus \{0\} \label{zz42}
\eeq
endowed with the symplectic form $\Om$ (\ref{zz43}).

Let us consider the functions
\mar{zz44,5}\ben
&& M_{12}=-M_(21)=q_1p_2-q_2p_1,  \label{zz44}\\
&& A_i=\op\sum_j M_{ij}p_j -\frac{q_i}{r}=q_ip^2 -p_i(p,q)-\frac{q_i}{r}, \qquad i=1,2,
\label{zz45}
\een
on the symplectic manifold $Z$ (\ref{zz42}). It is readily
observed that they are integrals of motion of the Hamiltonian $H$
(\ref{zz41}). One calls $M_{12}$ the angular momentum and $(A_i)$
the Rung--Lenz vector. Let us denote
\mar{zz50}\beq
M^2=(M_{12})^2, \qquad  A^2=(A_1)^2 + (A_2)^2=2M^2H+1.
\label{zz50}
\eeq

Let $Z_0\subset Z$ be a closed subset of points where $M_{12}=0$.
A direct computation shows that the functions $(M_{12},A_i)$
(\ref{zz44}) -- (\ref{zz45}) are independent on an open
submanifold
\mar{zz47}\beq
U=Z\setminus Z_0. \label{zz47}
\eeq
of $Z$. At the same time, the functions $(H,M_{12},A_i)$ are
nowhere independent on $U$ because it follows from the expression
(\ref{zz50}) that
\mar{zz52}\beq
H=\frac{A^2-1}{2M^2} \label{zz52}
\eeq
on $U$ (\ref{zz47}). The well known dynamics of the Kepler system
shows that the Hamiltonian vector field of its Hamiltonian is
complete on $U$ (but not on Z).

The Poisson bracket of integrals of motion $M_{12}$ (\ref{zz44})
and $A_i$ (\ref{zz45}) obeys the relations
\mar{zz56,7}\ben
&& \{M_{12},A_i\}=\eta_{2i}A_1 -\eta_{1i}A_2, \label{zz56}\\
&& \{A_1,A_2\}=2HM_{12}=\frac{A^2-1}{M_{12}}, \label{zz57}
\een
where $\eta_{ij}$ is an Euclidean metric on $\Bbb R^2$. It is
readily observed that these relations take the form (\ref{nc1}).
However, the matrix function $\bs$ of the relations (\ref{zz56})
-- (\ref{zz57}) fails to be of constant rank at points where
$H=0$. Therefore, let us consider the open submanifolds
$U_-\subset U$ where $H<0$ and $U_+$ where $H>0$. Then we observe
that the Kepler system with the Hamiltonian $H$ (\ref{zz41}) and
the integrals of motion $(M_{ij},A_i)$ (\ref{zz44}) --
(\ref{zz45}) on $U_-$ and the Kepler system with the Hamiltonian
$H$ (\ref{zz41}) and the integrals of motion $(M_{ij},A_i)$
(\ref{zz44}) -- (\ref{zz45}) on $U_+$ are superintegrable
Hamiltonian systems. Moreover, these superintegrable systems can
be brought into the form (\ref{zz60}) as follows.

Let us replace the integrals of motions $A_i$ with the integrals
of motion
\mar{zz61}\beq
L_i=\frac{A_i}{\sqrt{-2H}} \label{zz61}
\eeq
on $U_-$, and with the integrals of motion
\mar{zz62}\beq
K_i=\frac{A_i}{\sqrt{2H}} \label{zz62}
\eeq
on $U_+$.

The superintegrable system $(M_{12},L_i)$ on $U_-$ obeys the
relations
\mar{zz66,7}\ben
&& \{M_{12},L_i\}=\eta_{2i}L_1 -\eta_{1i}L_2, \label{zz66}\\
&& \{L_1,L_2\}=-M_{12}. \label{zz67}
\een
Let us denote $M_{i3}=-L_i$ and put the indexes
$\m,\nu,\al,\bt=1,2,3$. Then the relations (\ref{zz66}) --
(\ref{zz67}) are brought into the form
\mar{zz68}\beq
\{M_{\m\nu},M_{\al\bt}\}=\eta_{\m\bt}M_{\nu\al} +
\eta_{\nu\al}M_{\m\bt} -
\eta_{\m\al}M_{\nu\bt}-\eta_{\nu\bt}M_{\m\al} \label{zz68}
\eeq
where $\eta_{\m\nu}$ is an Euclidean metric on $\Bbb R^3$. A
glance at the expression (\ref{zz68}) shows that the integrals of
motion $M_{12}$ (\ref{zz44}) and $L_i$ (\ref{zz61}) constitute the
Lie algebra $\cG=so(3)$. Its corank equals 1. Therefore the
superintegrable system $(M_{12}, L_i)$ on $U_-$ is maximally
superintegrable. The equality (\ref{zz52}) takes the form
\mar{zz100}\beq
M^2 +L^2=-\frac1{2H}. \label{zz100}
\eeq

The superintegrable system $(M_{12},K_i)$ on $U_+$ obeys the
relations
\mar{zz76,7}\ben
&& \{M_{12},K_i\}=\eta_{2i}K_1 -\eta_{1i}K_2, \label{zz76}\\
&& \{K_1,K_2\}=M_{12}. \label{zz77}
\een
Let us denote $M_{i3}=-K_i$ and put the indexes
$\m,\nu,\al,\bt=1,2,3$. Then the relations (\ref{zz76}) --
(\ref{zz77}) are brought into the form
\mar{zz78}\beq
\{M_{\m\nu},M_{\al\bt}\}=\rho_{\m\bt}M_{\nu\al} +
\rho_{\nu\al}M_{\m\bt} -
\rho_{\m\al}M_{\nu\bt}-\rho_{\nu\bt}M_{\m\al} \label{zz78}
\eeq
where $\rho_{\m\nu}$ is a pseudo-Euclidean metric of signature
$(+,+,-)$ on $\Bbb R^3$. A glance at the expression (\ref{zz78})
shows that the integrals of motion $M_{12}$ (\ref{zz44}) and $K_i$
(\ref{zz62}) constitute the Lie algebra $so(2,1)$. Its corank
equals 1. Therefore the superintegrable system $(M_{12}, K_i)$ on
$U_+$ is maximally superintegrable. The equality (\ref{zz52})
takes the form
\mar{zz101}\beq
K^2 -M^2=\frac1{2H}. \label{zz101}
\eeq

Thus, the Kepler system on a phase space $\Bbb R^4$ falls into two
different maximally superintegrable systems on open submanifolds
$U_-$ and $U_+$ of $\Bbb R^4$. We agree to call them the Kepler
superintegrable systems on $U_-$ and $U_+$, respectively.

Let us study the first one. Put
\mar{zz102,j50}\ben
&& F_1=-L_1, \qquad F_2=-L_2, \qquad F_3=-M_{12}, \label{zz102}\\
&& \{F_1,F_2\}=F_3, \qquad \{F_2,F_3\}=F_1, \qquad \{F_3,F_1\}=F_2.
\label{j50}
\een
We have the fibered manifold
\mar{zz103}\beq
F: U_-\to N\subset\cG^*, \label{zz103}
\eeq
which is the momentum mapping to the Lie coalgebra
$\cG^*=so(3)^*$, endowed with the coordinates $(x_i)$ such that
integrals of motion $F_i$ on $\cG^*$ read $F_i=x_i$. A base $N$ of
the fibered manifold (\ref{zz103}) is an open submanifold of
$\cG^*$ given by the coordinate condition $x_3\neq 0$. It is a
union of two contractible components defined by the conditions
$x_3>0$ and $x_3<0$. The coinduced Lie--Poisson structure on $N$
takes the form
\mar{j51}\beq
w= x_2\dr^3\w\dr^1 + x_3\dr^1\w\dr^2 + x_1\dr^2\w\dr^3.
\label{j51}
\eeq

The coadjoint action of $so(3)$ on $N$ reads
\be
\ve_1=x_3\dr^2-x_2\dr^3, \qquad \ve_2=x_1\dr^3-x_3\dr^1, \qquad
\ve_3=x_2\dr^1-x_1\dr^2.
\ee
The orbits of this coadjoint action are given by the equation
\be
x_1^2 + x_2^2 + x_3^2={\rm const}.
\ee
They are the level surfaces of the Casimir function
\be
C=x_1^2 + x_2^2 + x_3^2
\ee
and, consequently, the Casimir function
\mar{zz120}\beq
h=-\frac12(x_1^2 + x_2^2 + x_3^2)^{-1}. \label{zz120}
\eeq
A glance at the expression (\ref{zz100}) shows that the pull-back
$F^*h$ of this Casimir function (\ref{zz120}) onto $U_-$ is the
Hamiltonian $H$ (\ref{zz41}) of the Kepler system on $U_-$.

As was mentioned above, the Hamiltonian vector field of $F^*h$ is
complete. Furthermore, it is known that invariant submanifolds of
the superintegrable Kepler system on $U_-$ are compact. Therefore,
the fibered manifold $F$ (\ref{zz103}) is a fiber bundle in
accordance with Proposition \ref{cmp15}. Moreover, this fiber
bundle is trivial because $N$ is a disjoint union of two
contractible manifolds. Consequently, it follows from Theorem
\ref{cmp36} that the Kepler superintegrable system on $U_-$ is
globally superintegrable, i.e., it admits global generalized
action-angle coordinates as follows.

The Poisson manifold $N$ (\ref{zz103}) can be endowed with the
coordinates
\be
(I,x_1,\g), \qquad I<0, \qquad \g\neq\frac{\pi}2,\frac{3\pi}2,
\ee
defined by the equalities
\mar{j52}\beq
I=-\frac12(x_1^2 + x_2^2 + x_3^2)^{-1}, \quad
x_2=(-\frac1{2I}-x_1^2)^{1/2}\sin\g, \quad
x_3=(-\frac1{2I}-x_1^2)^{1/2}\cos\g. \label{j52}
\eeq
It is readily observed that the coordinates (\ref{j52}) are the
Darboux coordinates of the Lie--Poisson structure (\ref{j51}) on
$U_-$, namely,
\mar{j53}\beq
w=\frac{\dr}{\dr x_1}\w \frac{\dr}{\dr \g}. \label{j53}
\eeq

Let $\vt_I$ be the Hamiltonian vector field of the Casimir
function $I$ (\ref{j52}). By virtue of Proposition \ref{nc8}, its
flows are invariant submanifolds of the Kepler superintegrable
system on $U_-$. Let $\al$ be a parameter along the flows of this
vector field, i.e.,
\mar{zz121}\beq
\vt_I= \frac{\dr}{\dr \al}. \label{zz121}
\eeq
Then $N$ is provided with the generalized action-angle coordinates
$(I,x_1,\g,\al)$ such that the Poisson bivector associated to the
symplectic form $\Om$ on $N$ reads
\mar{j54}\beq
W= \frac{\dr}{\dr I}\w \frac{\dr}{\dr \al} + \frac{\dr}{\dr x_1}\w
\frac{\dr}{\dr \g}. \label{j54}
\eeq
Accordingly, Hamiltonian vector fields of integrals of motion
$F_i$ (\ref{zz102}) take the form
\be
&& \vt_1= \frac{\dr}{\dr \g}, \\
&& \vt_2= \frac1{4I^2}(-\frac1{2I}-x_1^2)^{-1/2}\sin\g
\frac{\dr}{\dr \al} - x_1 (-\frac1{2I}-x_1^2)^{-1/2}\sin\g
\frac{\dr}{\dr \g} - \\
&& \qquad (-\frac1{2I}-x_1^2)^{1/2}\cos\g
\frac{\dr}{\dr x_1}, \\
&& \vt_3= \frac1{4I^2}(-\frac1{2I}-x_1^2)^{-1/2}\cos\g
\frac{\dr}{\dr \al} - x_1 (-\frac1{2I}-x_1^2)^{-1/2}\cos\g
\frac{\dr}{\dr \g} + \\
&&\qquad (-\frac1{2I}-x_1^2)^{1/2}\sin\g \frac{\dr}{\dr x_1}.
\ee
A glance at these expressions shows that the vector fields $\vt_1$
and $\vt_2$ fail to be complete on $U_-$ (see Remark \ref{zz90}).

One can say something more about the angle coordinate $\al$. The
vector field $\vt_I$ (\ref{zz121} reads
\be
\frac{\dr}{\dr\al}= \op\sum_i(\frac{\dr H}{\dr p_i}\frac{\dr}{\dr
q_i}-\frac{\dr H}{\dr q_i}\frac{\dr}{\dr p_i}).
\ee
This equality leads to the relations
\be
\frac{\dr q_i}{\dr \al}=\frac{\dr H}{\dr p_i}, \qquad \frac{\dr
p_i}{\dr \al}=-\frac{\dr H}{\dr q_i},
\ee
which take the form of the Hamilton equations. Therefore, the
coordinate $\al$ is a cyclic time $\al=t\,{\rm mod}2\pi$ given by
the well-known expression
\be
\al=\f-a^{3/2}e\sin(a^{-3/2}\f),\quad
r=a(1-e\cos(a^{-3/2}\f))\quad a=-\frac1{2I}, \quad
e=(1+2IM^2)^{1/2}.
\ee

Now let us turn to the Kepler superintegrable system on $U_+$. It
is a globally superintegrable system with noncompact invariant
submanifolds as follows.

Put
\mar{zz102a,j50a}\ben
&& S_1=-K_1, \qquad S_2=-K_2, \qquad S_3=-M_{12}, \label{zz102a}\\
&& \{S_1,S_2\}=-S_3, \qquad \{S_2,S_3\}=S_1, \qquad \{S_3,S_1\}=S_2.
\label{j50a}
\een
We have the fibered manifold
\mar{zz103a}\beq
S: U_+\to N\subset\cG^*, \label{zz103a}
\eeq
which is the momentum mapping to the Lie coalgebra
$\cG^*=so(2,1)^*$, endowed with the coordinates $(x_i)$ such that
integrals of motion $S_i$ on $\cG^*$ read $S_i=x_i$. A base $N$ of
the fibered manifold (\ref{zz103a}) is an open submanifold of
$\cG^*$ given by the coordinate condition $x_3\neq 0$. It is a
union of two contractible components defined by the conditions
$x_3>0$ and $x_3<0$. The coinduced Lie--Poisson structure on $N$
takes the form
\mar{j51a}\beq
w= x_2\dr^3\w\dr^1 - x_3\dr^1\w\dr^2 + x_1\dr^2\w\dr^3.
\label{j51a}
\eeq

The coadjoint action of $so(2,1)$ on $N$ reads
\be
\ve_1=-x_3\dr^2-x_2\dr^3, \qquad \ve_2=x_1\dr^3+x_3\dr^1, \qquad
\ve_3=x_2\dr^1-x_1\dr^2.
\ee
The orbits of this coadjoint action are given by the equation
\be
x_1^2 + x_2^2 - x_3^2={\rm const}.
\ee
They are the level surfaces of the Casimir function
\be
C=x_1^2 + x_2^2 - x_3^2
\ee
and, consequently, the Casimir function
\mar{zz120a}\beq
h=\frac12(x_1^2 + x_2^2 - x_3^2)^{-1}. \label{zz120a}
\eeq
A glance at the expression (\ref{zz101}) shows that the pull-back
$S^*h$ of this Casimir function (\ref{zz120a}) onto $U_+$ is the
Hamiltonian $H$ (\ref{zz41}) of the Kepler system on $U_+$.

As was mentioned above, the Hamiltonian vector field of $S^*h$ is
complete. Furthermore, it is known that invariant submanifolds of
the superintegrable Kepler system on $U_+$ are diffeomorphic to
$\Bbb R$. Therefore, the fibered manifold $S$ (\ref{zz103a}) is a
fiber bundle in accordance with Proposition \ref{cmp15}. Moreover,
this fiber bundle is trivial because $N$ is a disjoint union of
two contractible manifolds. Consequently, it follows from Theorem
\ref{cmp36} that the Kepler superintegrable system on $U_+$ is
globally superintegrable, i.e., it admits global generalized
action-angle coordinates as follows.

The Poisson manifold $N$ (\ref{zz103a}) can be endowed with the
coordinates
\be
(I,x_1,\la), \qquad I>0, \qquad \la\neq 0,
\ee
defined by the equalities
\be
I=\frac12(x_1^2 + x_2^2 - x_3^2)^{-1}, \quad
x_2=(\frac1{2I}-x_1^2)^{1/2}\cosh\la, \quad
x_3=(\frac1{2I}-x_1^2)^{1/2}\sinh\la.
\ee
These coordinates are the Darboux coordinates of the Lie--Poisson
structure (\ref{j51a}) on $N$, namely,
\mar{j53a}\beq
w=\frac{\dr}{\dr \la}\w \frac{\dr}{\dr x_1}. \label{j53a}
\eeq

Let $\vt_I$ be the Hamiltonian vector field of the Casimir
function $I$ (\ref{j52}). By virtue of Proposition \ref{nc8}, its
flows are invariant submanifolds of the Kepler superintegrable
system on $U_+$. Let $\tau$ be a parameter along the flows of this
vector field, i.e.,
\mar{zz121a}\beq
\vt_I= \frac{\dr}{\dr \tau}. \label{zz121a}
\eeq
Then $N$ (\ref{zz103a}) is provided with the generalized
action-angle coordinates $(I,x_1,\la,\tau)$ such that the Poisson
bivector associated to the symplectic form $\Om$ on $U_+$ reads
\mar{j54a}\beq
W= \frac{\dr}{\dr I}\w \frac{\dr}{\dr \tau} + \frac{\dr}{\dr
\la}\w \frac{\dr}{\dr x_1}. \label{j54a}
\eeq
Accordingly, Hamiltonian vector fields of integrals of motion
$S_i$ (\ref{zz102a}) take the form
\be
&& \vt_1= -\frac{\dr}{\dr \la}, \\
&& \vt_2= \frac1{4I^2}(\frac1{2I}-x_1^2)^{-1/2}\cosh\la
\frac{\dr}{\dr \tau} + x_1 (\frac1{2I}-x_1^2)^{-1/2}\cosh\la
\frac{\dr}{\dr \la} + \\
&& \qquad (\frac1{2I}-x_1^2)^{1/2}\sinh\la
\frac{\dr}{\dr x_1}, \\
&& \vt_3= \frac1{4I^2}(\frac1{2I}-x_1^2)^{-1/2}\sinh\la
\frac{\dr}{\dr \tau} + x_1 (\frac1{2I}-x_1^2)^{-1/2}\sinh\la
\frac{\dr}{\dr \la} + \\
&&\qquad (\frac1{2I}-x_1^2)^{1/2}\cosh\la \frac{\dr}{\dr x_1}.
\ee

Similarly to the angle coordinate $\al$ (\ref{zz121}), the
generalized angle coordinate $\tau$ (\ref{zz121a}) obeys the
Hamilton equations
\be
\frac{\dr q_i}{\dr \tau}=\frac{\dr H}{\dr p_i}, \qquad \frac{\dr
p_i}{\dr \tau}=-\frac{\dr H}{\dr q_i}.
\ee
Therefore, it is the time $\tau=t$ given by the well-known
expression
\be
\tau=s-a^{3/2}e\sinh (a^{-3/2}s),\quad r=a(e\cosh
(a^{-3/2}s)-1)\quad a=\frac1{2I}, \quad e=(1+2IM^2)^{1/2}.
\ee

\end{document}